\documentclass{sig-alternate-10pt}
\usepackage{color}
\usepackage{subfigure}

\def\P{{\mathbb P}}
\def\E{{\mathbb E}}

\def\E{\mathbb E} 

\newcommand{\mb}{\mathbf}

\DeclareMathOperator*{\argmin}{arg\,min}

\newtheorem{proposition}{Proposition}
\newtheorem{heuristic}{Algorithm}

\graphicspath{{./pics/}{./}}

\begin{document}

\title{A State-Space Approach for Optimal Traffic Monitoring \\via Network Flow Sampling}

\numberofauthors{3} 
%
\author{
%
%
\alignauthor
Michael Kallitsis\\
       \affaddr{Merit Network, Inc.}\\
        \affaddr{Ann Arbor, MI}\\
       \email{mgkallit@merit.edu}
\alignauthor
Stilian Stoev\\
       \affaddr{Department of Statistics}\\
       \affaddr{University of Michigan}\\
        \affaddr{Ann Arbor, MI}\\
       \email{sstoev@umich.edu}
\alignauthor
George Michailidis\\
       \affaddr{Department of Statistics}\\
       \affaddr{University of Michigan}\\
        \affaddr{Ann Arbor, MI}\\
       \email{gmichail@umich.edu}       
}
\date{29 April 2013}

\maketitle
\begin{abstract}
The robustness and integrity of IP networks require efficient tools for traffic monitoring and analysis, which scale well with traffic volume and network size. We address the problem of optimal large-scale \emph{flow} monitoring of computer networks under resource constraints. We propose a stochastic optimization framework where traffic measurements are done by exploiting the spatial (across network links) and temporal relationship of traffic flows.
Specifically, given the network topology, the state-space characterization of network flows and sampling constraints at each monitoring station, we seek an optimal packet sampling strategy that yields the ``best" traffic volume estimation for all flows of the network. 
The optimal sampling design is the result of a concave minimization problem; then, Kalman filtering is employed to
yield a sequence of traffic estimates for each network flow. We evaluate our algorithm
using real-world Internet2 data.
\end{abstract}




\section{Introduction}

Advances in networking technologies and high performance computing have led to an unprecedented growth of a vast array of applications such as cloud computing, social networking, video on demand, cloud storage, and voice over IP, to name a few. At the same time, malicious network activity  remains a big concern since network attacks become more sophisticated. Therefore, it is extremely important for network operators to have an accurate global-view of their network for diagnosing anomalous activity~\cite{Lakhina:2004:DNT:1030194.1015492}, for optimal network capacity planning and  quality of service considerations~\cite{Claffy:1993:ASM:167954.166256}.  These can be achieved through network monitoring. However, monitoring everywhere and constantly is expensive, energy inefficient and computationally challenging. Thus, one should employ statistical tools for traffic estimation through limited collection of measurements.

Network monitoring has traditionally been  done with SNMP measurements~\cite{Claffy:1993:ASM:167954.166256, Papagiannaki:2004:DAM:1028788.1028808, Zhang:2003:IAT:863955.863990}. SNMP measurements provide link counts which give the aggregate traffic volume at the observation point of interest. Recently, more granularity can be achieved by performing flow-level measurements using tools such as Cisco's NetFlow. The latter approach simplifies the monitoring task significantly. The idea is to sample packets from flows of interests at specific router interfaces, henceforth called observation points. For each packet sampled, several header information can be extracted and recorded for further analysis.  Each packet from a flow (a flow can be an aggregate flow, i.e., flows originating from a particular subnet or an autonomous system) is sampled independently with a particular \emph{sampling probability} (also known as \emph{sampling rate}). Typical sampling rates are between $0.01$ (i.e., only 1 out of 100 packets is selected for sampling) and $0.20$. Higher sampling rates can also be chosen, but they amount to valuable resource consumption at each router (cache memory, CPU cycles, storage, network bandwidth and power). Thus, judicious choice of the sampling rates greatly affects the efficient operation of the network. 

Regardless of the measurement technique, network monitoring aims to several objectives: a) identification of the traffic volume for network flows (known as \emph{traffic matrix})~\cite{Singhal:2008:OSS:1375457.1375474, Medina:2002:TME:633025.633041, Papagiannaki:2004:DAM:1028788.1028808, Zhang:2003:IAT:863955.863990,   Soule:2005:TMB:1064212.1064259, 1424313}, i.e., traffic for each origin-destination pair given the link counts and the topology of the network, b) identification of flow characteristics such as end-to-end network delay~\cite{4016134, 5873177, Coates:2007:CNM:1298306.1298340}, flow length, flow distribution or other flow statistics~\cite{4215789, Duffield02propertiesand, Duffield:2005:EFD:1103543.1103544}. To accomplish these, several interesting problems arise, such as the \emph{subset selection} problem for choosing the locations of the monitoring stations~\cite{4016134, Coates:2007:CNM:1298306.1298340} and the sampling design problem~\cite{Singhal:2008:OSS:1375457.1375474}.

This paper aims to address the problem of flow estimation through an optimal sampling strategy under resource constraints (see also~\cite{Duffield:2004:FSU:1005686.1005699}). 
We assume that network monitoring is performed via Netflow-alike measurements.
Our framework takes into account both temporal correlations of the flows (see also~\cite{1424313}) as well as their spatial correlation~\cite{Duffield:2005:OCS:1251086.1251094}.
We adapt  Bernoulli sampling for our measurements at each observation site (that is, the information of a packet at each network link is recorded according to the sampling rate/probability described above), but alternative sampling techniques
can aslo be considered~\cite{Cohen:2007:AEA:1298306.1298344, Cohen:2012:DLN:2254756.2254798, Duffield:2012:FSA:2318857.2254800}.

The toy-example of Figure~\ref{fig:Internet2} provides more insights on the proposed method;  assume we have flows between each network node.
Further, assume  flow-monitoring tools can sample with rate $20$ out of every $100$ packets.
How should the network operator assign the sampling rates to each flow subject to the sampling capacity of each link?
By considering the network topology, one would expect that ``long-flows"
do not require many samples at every single link they traverse. For example, the flow from 'Houston' to 'NY' needs not be sampled at every link on its path. 
Sampling on link 'Houston' - 'Kansas' may be sufficient; this will leave the resources of
the intermediate link
to be utilized for monitoring the ``short" flow that traverses \emph{only} the link 'Kansas' to 'Chicago'.
Similarly, a stochastic characterization of the ``evolution"  of each flow over time provides valuable information for choosing the ``'best" sampling strategy. Section~\ref{sec:description}
unifies these ideas in a stochastic optimization framework.

Similar approaches for flow estimation via state-space models have also appeared in~\cite{Singhal:2008:OSS:1375457.1375474, Soule:2005:TMB:1064212.1064259}. 
However, in~\cite{Singhal:2008:OSS:1375457.1375474} the proposed method addresses flow estimation on a \emph{single} network link
and the spatial correlation between flows are not examined. In~\cite{Soule:2005:TMB:1064212.1064259}, the suggested state-space model
considers link counts at \emph{every} network link without relying on flow sampling (in particular, SNMP counts are considered). In large-scale networks this approach is not practical nor feasible.
Further, compared to~\cite{Soule:2005:TMB:1064212.1064259}, we study a framework better-tailored to the ongoing measurement process.

The contributions of this paper are twofold: 
a) We present a stochastic optimization framework for finding the optimal sampling design that would yield the best traffic estimates for each flow (section~\ref{sec:description}).
The model views each flow as a stochastic process.  The \emph{state} of the system at a particular time instance is the volume of traffic that each flow carries at that instance.
Through sampling, we get a \emph{partially observed} system; this observation uncertainty is captured through the measurement equations we define next.
The goal is to find the ``best" sampling strategy that minimizes the estimation error over the (finite) horizon of interest. 
This is the first attempt to model the flow estimation problem under a stochastic control framework;
b) We study an approximation scheme for the solution of the above-mentioned stochastic optimization problem (section~\ref{sec:solution}). The problem of obtaining the optimal sampling rates is then reduced to a \emph{deterministic} optimization problem that can be solved \emph{a priori}. Based on the calculated sampling rates, traffic estimation for each time-step is then performed via the Kalman filter.
As illustrated in Figure~\ref{fig:opt_vs_naive1} the proposed approach poses significant gains over existing techniques.
We evaluate our approach using real-world data obtained from Internet2 (section~\ref{sec:peva:sampling}).

\begin{figure}[tp]
\centering
\includegraphics*[bb=155 60 835 530, scale=0.30]{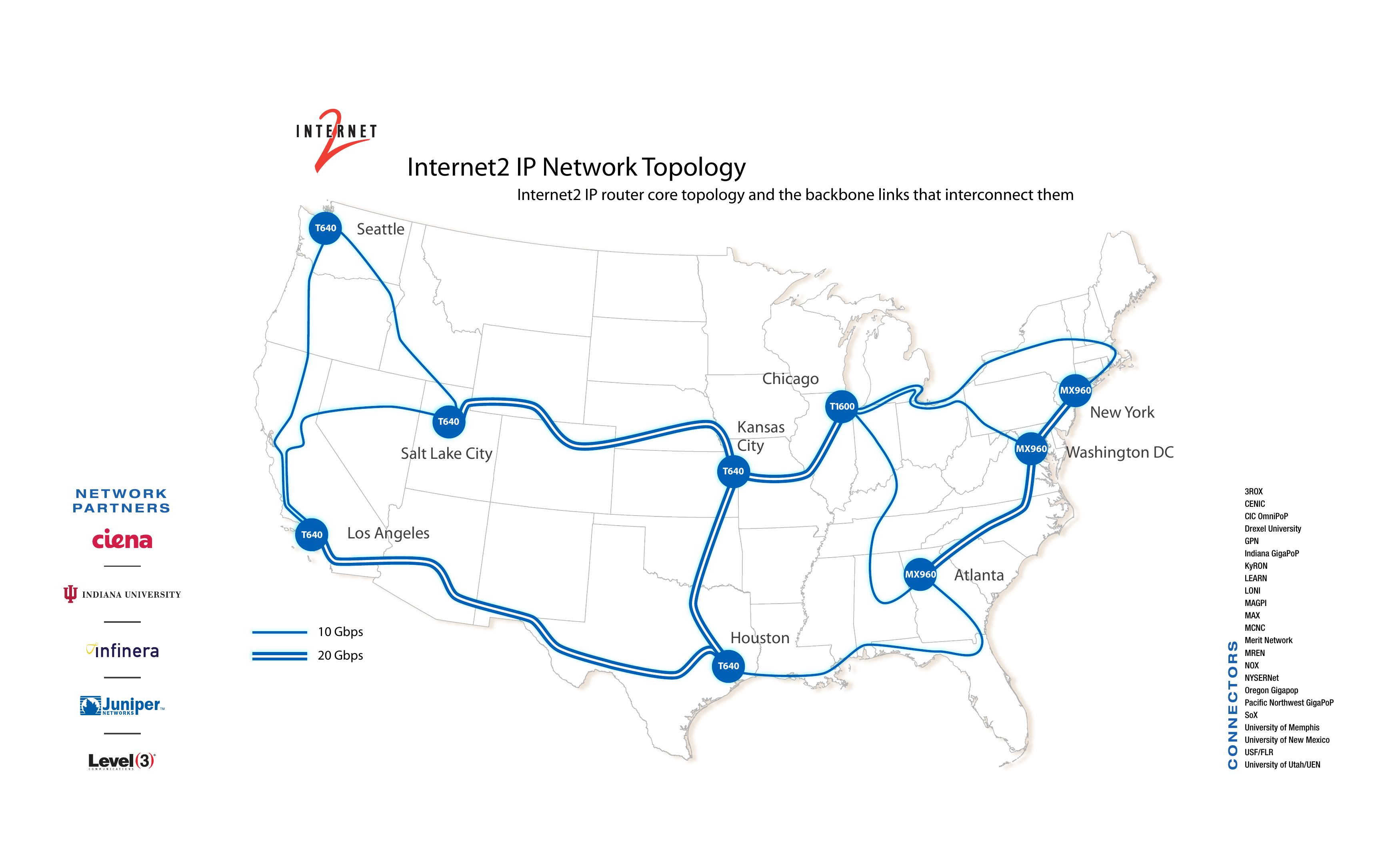}  
\caption{Backbone network of Internet2.}
\label{fig:Internet2}
\end{figure}

\section{Problem Description}
\label{sec:description}

Consider a communication network
of $N$ nodes and $L$ links. The total number of traffic flows, i.e. source and 
destination $(S,D)$ pairs, is denoted by $J$.  
We denote the set of flows as $\mathcal{J} := \{1,2,\ldots, J\}$ and the set of all network links with $\mathcal{L}:=\{1, 2,\ldots, L\}$. 
Traffic is routed over the network along predefined paths described 
by a routing matrix $R = (r_{\ell, j})_{L\times J},$ with $r_{\ell,j} =  1, \mbox{when route $j$ uses link $\ell$ }$ and  $0$ otherwise.   
 Let  $$\mb{x}_t= (x_t(1), x_t(2), \ldots, x_t(J))^T , t=1,2,\ldots$$ 
and $$\mb{y}_t= ( y_t(1), y_t(2) , \ldots, y_t(L)  )^T, t=1,2,\ldots$$
 be the vector time series\footnote{Here, time is discrete and traffic loads 
are measured in bytes or packets per unit time, over a time scale greater than the round-trip time of the network.}
of traffic traversing all $J$ routes and $L$ links, respectively. We shall ignore network delays and adopt the 
assumption of \emph{instantaneous propagation}. This is reasonable when traffic is monitored at 
a time-scale coarser than the round-trip time of the network, which is the case in our setting.
We thus obtain that the link and route level traffic are related through the fundamental {\em routing equation}\footnote{
Note that an in backbone IP networks the routing matrix $R$ does not change often.}
\begin{equation}\label{e:routing} 
\mb{y}_t = R \mb{x}_t.
\end{equation} 

The spatial correlation between flows encoded in the routing matrix $R$, will play an important role in determining
the optimal sampling design. Before discussing the solution of our optimization problem though, we first consider in detail
all the components of our stochastic control formulation.

\begin{figure}[tp]
\centering
\includegraphics*[scale=0.4]{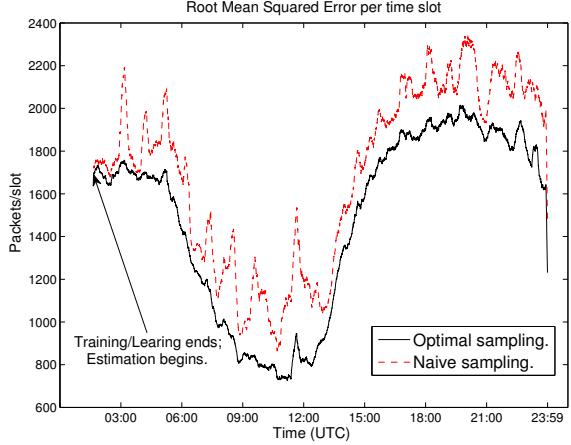}   
\caption{Estimation error per time interval: comparison of optimal  versus na\"ive sampling for the $72$ Internet2 flows on 2009-03-17. }
\label{fig:opt_vs_naive1}
\end{figure}


\subsection{A State-Space Model}

In this paper, we model the evolution over time of the volume of each flow as a stochastic process. In particular,
we model the dynamics of each flow $j, j=1,2,\ldots, J$, as the following Markov process:

\begin{equation}
\label{eq:stateeq}
 x_{t+1}(j) = \rho_j x_t(j) + w_t(j), t=0, 1,2, \ldots
\end{equation}

where $x_t(j)$ represents the state of flow $j$ at time $t$, namely the numbers of packets (or bytes)  carried at time interval $t$. For the purposes of this
paper we assume that each time interval has a duration of $10$ minutes. The sequence of random variables $w_t(j), t=1, 2, \ldots$ represent random noise.
They belong to the set of \emph{primitive random variables}, meaning that they are mutually independent. They are also independent from the
state random variables. We assume   noise   to be Gaussian with zero mean and variance equal to $\sigma_{\epsilon}^2$.
To fully characterize the system
evolution  for flow $j$ of Eq.~\eqref{eq:stateeq} we need   initial state $x_0(j)$, which is also assumed to be
Gaussian. Its mean and variance can be calculated during a calibration phase. To summarize, for $t=0, 1,2, \ldots$, the \emph{system dynamics} are described by
\begin{equation}
 \mb x_{t+1} = F \mb x_t + \mb w_t, 
\end{equation}
with $\mb x_t$ being the vector representing the ``state" of each flow at time $t$, and $F$ a diagonal matrix of the coefficients $\rho_j$. Moreover,
we have the following  probability density functions for the primitive random variables described above, i.e.,
\begin{align}
 p(\mb x_0) & = c_1 {\rm exp} \{   -\frac{1}{2} [  (\mb x_0 - \mb {\bar x_0} )^T ( \hat{P}_{0|-1}  )^{-1} (\mb x_0 - \mb {\bar x_0} )  ]  \}  \label{eq:prior:state}\\
 p(\mb w_t) & = c_2 {\rm exp} \{   -\frac{1}{2} [  \mb w_t^T  \hat{Q}^{-1}  \mb w_t  ] \}.  
\end{align}
The parameters $\mb x_0$, $\hat{P}_{0|-1}$, $F$ and ${\hat{Q}}$ can be determined through a short calibration phase
using techniques for fitting autoregressive models (see~\cite{Brockwell:1986:TST:17326}, Chapter 8).

\subsection{Traffic Measurement}

\subsubsection{Bernoulli flow sampling}

As mentioned above, at time $t$ the volume
of flow $j$ is denoted by the state variable $x_t(j)$. 
We adopt a Bernoulli sampling scheme~\cite{Duffield:2005:EFD:1103543.1103544}. This says that each packet of flow $j$ passing through 
the observation point $\ell$ at time $t$, is sampled with probability $u_t (\ell, j)$. 
In other words, the variable $u_t(\ell, j)$ specifies the sampling rate
at link $\ell$ for flow $j$ at  time $t$. 

The number of packets captured at observation point $\ell$ for
flow $j$ is given by the random variable $\tilde{y}_t(\ell, j)$. Given $x_t(j)$, $\tilde{y}_t(\ell, j)$ follows  a binomial distribution, i.e.,
\begin{equation}
\label{eq:bernoulli}
 \tilde{y}_t(\ell, j) \sim  \rm{Binomial}(x_t(j), u_t (\ell, j)).
\end{equation}
Based on the observations $\tilde{y}_t(\ell, j)$, the unbiased estimator for $x_t(\ell, j)$ -- the traffic volume at link $\ell$ for flow $j$ -- is given by:
\begin{equation}
\label{eq:link:est}
 \tilde{z}_{t}(\ell, j) = \frac{\tilde y_t(\ell ,j)}{u_t (\ell, j)} .
\end{equation}

The variance of the estimator at link $\ell$ for flow $j$ equals
\begin{align}
 \label{eq:var_ij}
  v_t(\ell, j) := &\E  [  (\tilde{z}_{t}(\ell, j) - \E \tilde{z}_{t}(\ell, j) )^2 | { x_t (j)} ] \notag\\
   = &\E [  (\tilde{z}_{t}(\ell, j) -  x_t(j) )^2  | x_t(j)] \notag\\
   = &x_t(j)   \frac{ 1 - u_t (\ell, j)} {u_t (\ell, j)} .
\end{align}

\subsubsection{Spatial combination of estimators}

We seek a combined estimator for the volume of flow $j$ that uses measurements from several observation points~\cite{Duffield:2005:OCS:1251086.1251094}.
Such an estimator can be expressed as,
\begin{equation}
\label{eq:partial:obs}
 z_t (j) = \sum_{\ell \in \ell(j)} w_{\ell, j} \tilde{z}_{t}(\ell, j),
\end{equation}
where $\ell(j)\subseteq \mathcal{L}$ is the set of links that flow $j$ traverses. 
Conditioned on the state $x_t(j)$, the observations at different links $\tilde{y}_{t}(\ell, j)$ are independent so one can calculate the variance
of the combined estimator to be
\begin{equation}
 \label{eq:varXgiven_xi}
 \rm{Var}( z_t (j) | x_t (j) ) = \sum_{\ell \in \ell(j)} {w^2_{\ell, j}}  v_t (\ell, j),
\end{equation}
where $v_t(\ell, j) =  x_t(j)   \frac{ 1 - u_t (\ell, j)} {u_t (\ell, j)} $ (see Eq.~\eqref{eq:var_ij}) and $\ell(j)$ is the set of links that flow $j$ is traversing and can be acquired from the routing matrix $R$. 
To obtain the \emph{best linear unbiased estimator} for all $j\in\mathcal{J}$, 
we want to find the optimal weights $w_{\ell, j}$ that 
minimize the above variance subject to
$\sum_{\ell \in\ell(j)} w_{\ell, j} =1$. Taking the Lagrangian and using the first-order optimality conditions we arrive to
\begin{equation}
\label{eq:weights}
 w_{\ell, j} = \frac{v_t(\ell, j)^{-1}}{\sum_{k\in \ell(j)} v_t (k, j)^{-1}}.
\end{equation}
Continuing from~\eqref{eq:varXgiven_xi},
\begin{equation}
 \label{eq:variances_Xi}
 \rm{Var}( z_t (j) | x_t (j) ) = \frac{1}{\sum_{k\in \ell(j)} \frac{1} { v_t (k, j)}  }.
\end{equation}

\subsubsection{The Measurement Equation}

From Eqs.~\eqref{eq:bernoulli},~\eqref{eq:link:est},~\eqref{eq:partial:obs} we see that we have a partially observable system. In other words,
the state $\mb x_t$ of the system -- the traffic volume for each flow at $t$ -- is not directly available, but can be inferred 
through the \emph{observations} $\mb z_t$. Using the normal approximation to the binomial distribution
 we get the following relation between the state and observations for flow $j$, for $t=1,2, \ldots$
\begin{equation}
\label{eq:measu:eq:flowj}
 z_t (j) = x_t(j) + r_t(j),
\end{equation}
where $r_t(j)$ is a Gaussian random variable, i.e. (see Eq.~\eqref{eq:variances_Xi})
\begin{align}
\label{eq:measu:gauss}
 r_t(j) 
          &\sim \mathcal{N}\Big(0, \frac{x_t(j)}{\sum_{k\in \ell(j)} \frac{1} { \frac{1}{u_t (k, j)} - 1 }  }\Big)
\end{align}

The measurement equation, for all flows becomes
\begin{equation}
\mb z_t = \mb x_t + \mb r_t,
\end{equation}
with the probability density function for $\mb r_t$ being
\begin{align}
 p(\mb r_t) & = c_3 {\rm exp} \{   -\frac{1}{2} [  \mb r_t^T  \hat{R}_{t}(\mb x_t, \mb u_t)^{-1}  \mb r_t  ] \},  
\end{align}
where $\hat{R}_{t}(\mb x_t, \mb u_t)$ is a covariance matrix. Using the proposed measurement scheme,
the covariance matrix is just a diagonal matrix with elements the variances shown in~Eq.~\eqref{eq:measu:gauss}.

\subsection{The Instantaneous Cost}

Let $\mb u_t \in (0, 1)^{L\times J}$ be the 
\emph{sampling matrix} arranged in a \emph{vector} form; the variable $u_t(\ell, j)$ specifies the sampling rate
at link $\ell$ for flow $j$ at  time $t$.  We \emph{define} the instantaneous cost  to be the estimation error at time $t$ as follows:
\begin{equation}
\label{eq:instcost1}
 c_t ({ \mb x_t}, \mb u_t) := {\rm {trace }} \Big(\E [ ( \mb z_t  - \mb x_t ) ( \mb z_t  - \mb x_t )^T]\Big),
\end{equation}
where $\mb z_t := \mb z_t (\mb u_t) $ is the vector of volume estimation for each flow $j = 1,2, \ldots J$ at time $t$ (see Eq.~\eqref{eq:partial:obs} and~\eqref{eq:weights}).

The instantaneous cost of~\eqref{eq:instcost1} can then be written as: 
\begin{equation}
\label{eq:instcost2}
c_t ({ \mb x_t}, \mb u_t) = \sum_{j=1}^J \frac{1}{\sum_{k\in \ell(j)} \frac{1} { v_t (k, j)}  }.
\end{equation}

\begin{proposition}
The instantaneous cost function shown in~\eqref{eq:instcost2} is concave in $\mb{u_t}$.
\end{proposition}
\begin{proof}
 Using the expanded form of Eq.~\eqref{eq:instcost2}  with $v_t (k, j) = x_t(j) (\frac{1}{u_t(k, j)} - 1)$ we observe that this
resembles the harmonic average of the terms $v_t (k, j)$. Using the fact that the harmonic average function is a concave and non-decreasing function~\cite{boyd:convex}, one
can easily verify that our objective function is concave in $\mb{u_t}$ as a composition of a concave and non-decreasing function with a convex function.
\end{proof}





\section{Optimal Sampling}
\label{sec:solution}

The problem at hand belongs to the category of measurement adaptive problems~\cite{1450418}.
In the general case,  the problem of optimal measurement control (see also sequential design of experiments~\cite{Robbins1952})
can be formulated as the following \emph{discrete-time}, \emph{finite-horizon}, \emph{partially-observable},
\emph{perfect-recall} stochastic control problem (see~\cite{1450418, Kumar:1986:SSE:40665}). 

We are given, the \emph{system evolution} equation, written as
\begin{equation}
\label{eq:system:evol}
 \mb x_{t+1} =  f_t ( \mb x_t, \mb w_t), t=0,1,2, \ldots, T, 
\end{equation}
the \emph{measurement equation}
\begin{equation}
\label{eq:measurement:eq}
 \mb z_{t} =  h_t ( \mb x_t, \mb u_t, \mb r_t), t=0,1,2, \ldots, T, 
\end{equation}
and the probability densities for the \emph{random variables} $\mb x_0$, $\mb w_t$ and $\mb r_t$
\begin{equation}
\label{eq:primitive}
  p(\mb x_0), p(\mb w_t), p(\mb r_t),  t=0,1,2, \ldots, T.
\end{equation}
The performance criterion is the expected cost over the horizon of interest
\begin{equation}
\label{eq:criterion}
V = \E\{  \sum_{t=0}^{T-1} c_t (\mb x_t, \mb u_t) + c_T (\mb x_T, \mb u_T) \},
\end{equation}
where $c_t (\mb x_t, \mb u_t)$ is the \emph{instantaneous cost} function and the expectation is taken with respect to the random variables $\mb x_t$. 

The problem is to find the \emph{optimal sampling strategy} 
$$g:=(g_1(\mb Z^1), \ldots, g_t(\mb Z^t), \ldots, g_T(\mb Z^T))$$ 
that minimizes the expected cost~\eqref{eq:criterion} over the horizon of interest subject to ``budgetary'' sampling constraints. The symbol
\begin{equation}
\label{eq:symbolhist}
\mb Z^t := ( \mb z_1, \mb z_2, \dots, \mb z_t ),
\end{equation}
represents the \emph{history} of observations up to time $t$. Simi-larly, the history of sampling rates up to time $t$ will be denoted as $\mb u^t$. As mentioned above, we assume a system with \emph{perfect-recall}
which means that all this information is available. Having calculated an optimal sampling strategy, the optimal sampling action
at time instance $t$ would be
\begin{equation}
\mb u_t = g_t(\mb Z^t).
\end{equation}
In other words, given the history of observations, the optimal action would at time $t$ will be given by the \emph{pre-calculated} optimal policy.

In the general case, applying dynamic programming
is hindered by the \emph{curse of dimensionality}~\cite{powellbook}.
Therefore, some sort of approximation techniques need to be involved~\cite{powellbook, Bertsekas:1996:NP:560669}.
Indeed, under the following conditions, the stochastic control problem can be solved efficiently~\cite{1098668} 
by exploiting the \emph{two-way separation} between estimation and control~\cite{1450418, Kumar:1986:SSE:40665}. The conditions\footnote{The models presented in~\cite{1098668, 1450418, Kumar:1986:SSE:40665}  cover more general cases than the one presented here. Specifically, in the general model the state of the system needs to be controlled as well, and a quadratic cost is associated with the system state.
Further, a quadratic cost in the decision variables may also exist.} are:
a) The system evolution equation (see Eq.~\eqref{eq:system:evol}) is linear;
b) The measurement equation (see Eq.~\eqref{eq:measurement:eq}) is linear in the state and measurement noise;
c) The primitive random variables are Gaussian;
and d)The instantaneous cost (in our case given by Eq.~\eqref{eq:instcost1}) is independent of the state $\mb x_t$.

In the special case we have a state equation of the following form
\begin{equation}
\mb x_{t+1} = F_t \mb x_t + \mb w_t,
\end{equation}
a measurement equation of this form
\begin{equation}
\mb z_{t} = H_t (\mb u_t) \mb x_t + \mb r_t,
\end{equation}
where $H_t(\mb u_t)$ relates the measurement matrix with the measurement control.
The probability density functions of the primitive random variables are
\begin{align}
 p(\mb x_0) & = c_1 {\rm exp} \{   -\frac{1}{2} [  (\mb x_0 - \mb {\bar x_0} )^T ( \hat{P}_{0|-1}  )^{-1} (\mb x_0 - \mb {\bar x_0} )  ]  \}  \\
 p(\mb w_t) & = c_2 {\rm exp} \{   -\frac{1}{2} [  \mb w_t^T  \hat{Q}_{t}^{-1}  \mb w_t  ] \}.  \\
  p(\mb r_t) & = c_3 {\rm exp} \{   -\frac{1}{2} [  \mb r_t^T  \hat{R}_{t}( \mb u_t)^{-1}  \mb r_t  ] \},  
\end{align}
where $ \hat{R}_{t}( \mb u_t)$ gives the relationship between measurement noise and sampling rate.
The  performance criterion is
\begin{equation}
 V = \E \{  \sum_{t=0}^T  c_t(\mb u_t)  \} =  \sum_{t=0}^T  c_t(\mb u_t)  
\end{equation}
subject to constraints on $\mb u_t$.

Given the above conditions and the sampling rates $\mb u_t$, traffic volume estimation can be performed
with the \emph{Kalman filter}~\cite{citeulike:347166}. $\hat{\mb x}_{t|t}$, the optimal estimate conditioned on $\mb Z^t$, is given by
\begin{equation}
\label{eq:kalman}
\hat{\mb x}_{t|t} = F_{t-1} \hat{\mb x}_{t-1|t-1} + \hat{K}_t [ \mb z_t - H_t F_{t-1} \hat{\mb x}_{t-1|t-1}  ], 
\end{equation} 
where $ \hat{K}_t $, the Kalman gain, is
\begin{equation}
 \label{eq:kalmangain}
  \hat{K}_t  = \hat{P}_{t|t-1} H_t^T ( H_t \hat{P}_{t|t-1}H_t^T + \hat{R}_t )^{-1}
\end{equation}
and $\hat{P}_{t|t}$, the conditional covariance of the error in the estimate of $\mb x_t$ given $\mb Z^t$ can be calculated recursively for $t=0,1, \ldots, T$ by
\begin{align}
\label{eq:kalmangain2}
 &\hat{P}_{t|t} = \hat{P}_{t|t-1} - \hat{P}_{t|t-1} H_t^T ( H_t \hat{P}_{t|t-1}H_t^T + \hat{R}_t )^{-1} H_t \hat{P}_{t|t-1} \notag\\
 & \hat{P}_{t|t-1} = \hat{Q}_{t-1} + F_{t-1} \hat{P}_{t-1|t-1}F_{t-1}^T.
\end{align}

The computation of the optimal sampling rates can be determined \emph{a priori} by calculating the solution
of the following nonlinear, deterministic control problem:
\begin{equation}
\label{eq:opt}
V^* = \min_{\mb u_t} \sum_{t=0}^T c_t( \hat{\mb x}_{t|t}, \mb u_t),
\end{equation}
subject to the ``budgetary" sampling constraints. $ \hat{\mb x}_{t|t}$ is the ``best" state estimator available
at the time the optimization problem is solved.

 The optimization problem~\eqref{eq:opt} 
 can be decomposed into a sequence of problems. For time $t$, and given the state estimation  $\hat{\mb x}_{t|t}$ we have:
\begin{align}
\label{eq:optproblem}
 \mb{u}_t^* \in &\argmin_{u_t{(\ell, j)} } c_t ({  \hat{\mb x}_{t|t}}, \mb u_t) \\
 &  {\text{ s.t. }} B\mb{u_t} \leq \mb{d} \notag,
\end{align}
where $\mb{u_t}$ is the vector of sampling rates, $B\mb{u}_t \leq \mb{d}$ represent linear ``budgetary'' constraints
per link, and B is a matrix of appropriate dimensions (deduced from the routing matrix R). 
The concavity of our objective function, along with the linearity of our constraints lead us to a \emph{minimization}
of a concave function. This is an NP-complete problem, known as \emph{global concave minimization}.

The solution of the concave program always lies on the \emph{vertices} of the convex hull defined
by the convex polyhedron of our linear ``budgetary'' inequalities shown in Eq.~\eqref{eq:optproblem}. 
The proof can be found in~\cite{pardalos:global}. 
The above proposition suggests that one way to solve our concave program -- but certainly not the most efficient one -- is to enumerate all the vertices
of the induced convex hull, and pick the one that yields the lowest error. 
More sophisticated methods for solving concave programs can be found in~\cite{pardalos:global, tuy1980concave}. 
The complete traffic estimation algorithm is presented next.

\begin{heuristic}[Optimal Sampling]
\label{alg:kalman}
$$\newline$$
\begin{enumerate}

\item For $t=1, \ldots, t_0$ collect traffic data to calibrate the model; i.e., find $\mb x_0$, $\hat{P}_{0|-1}$, $F$ and ${\hat{Q}}$.

\item For $t=t_0$, set $ \hat{\mb x}_{t|t}=  \mb x_0$, and solve~\eqref{eq:opt}. We have now obtained $\mb u_t$, for $t=t_0, \ldots, T$.
 
\item Using the optimal sampling rates of Step 2) and Eqs.~\eqref{eq:kalmangain} and~\eqref{eq:kalmangain2} calculate the Kalman gain.

\item Using Eq.~\eqref{eq:partial:obs} get the combined observation for each flow $j$.

\item With the observations acquired from~\eqref{eq:partial:obs}, use the Kalman filter~\eqref{eq:kalman} to obtain the traffic volume estimation for time $t$, given the past of observations.

\item Set $t=t+1$ and go to Step 3). Repeat until $t=T$.

\end{enumerate}
 
\end{heuristic}

\section{Performance Evaluation}
\label{sec:peva:sampling}


\begin{figure*}[tp]
\vspace{-35pt}
  \begin{center}
   \hspace{-55pt}
    \subfigure[Largest flow.]{\label{fig:opt_vs_naive_flow19}\includegraphics[scale=0.30]{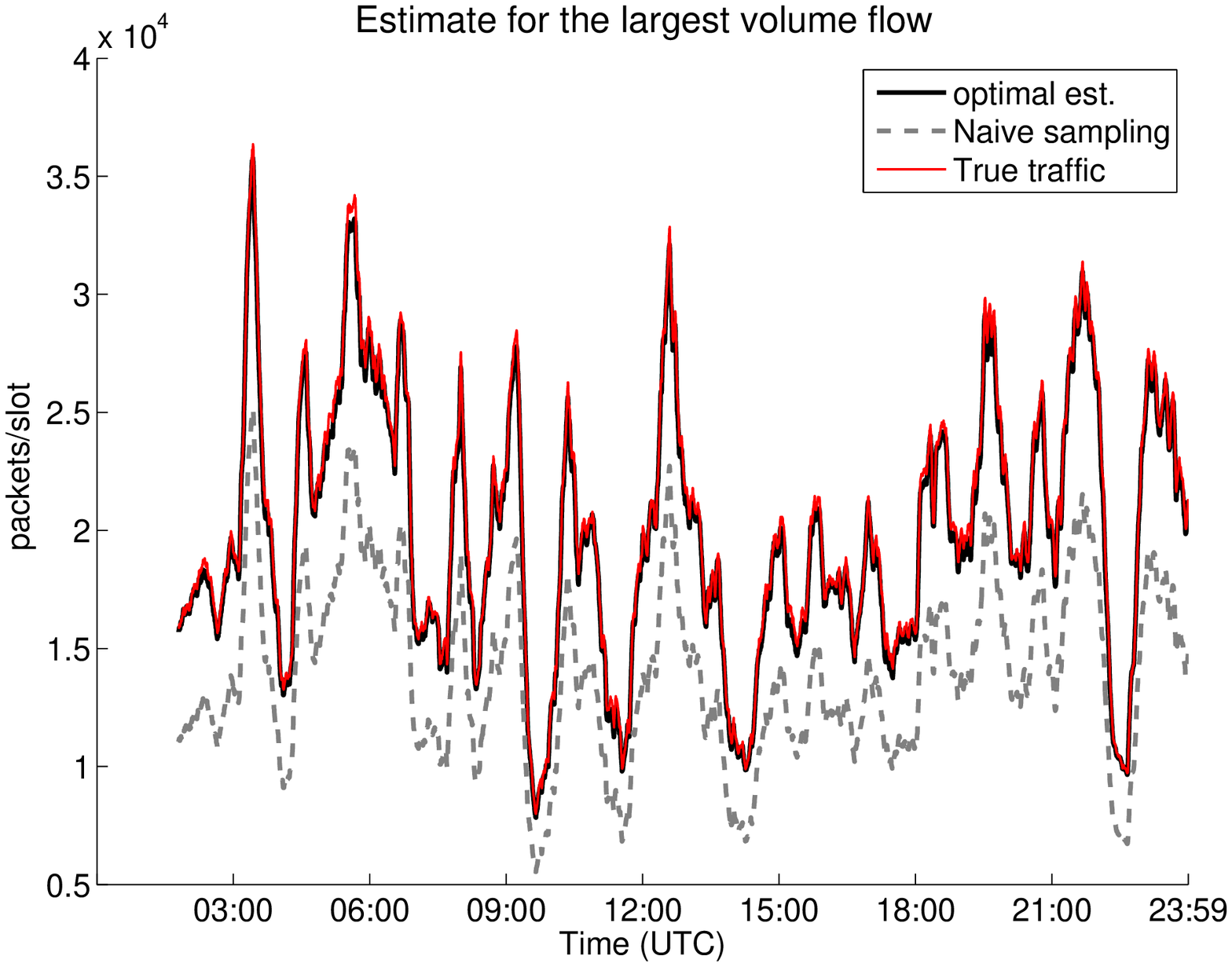}}
    \hspace{-15pt}
    \subfigure[Second largest flow.]{\label{fig:opt_vs_naive_flow66}\includegraphics[scale=0.30]{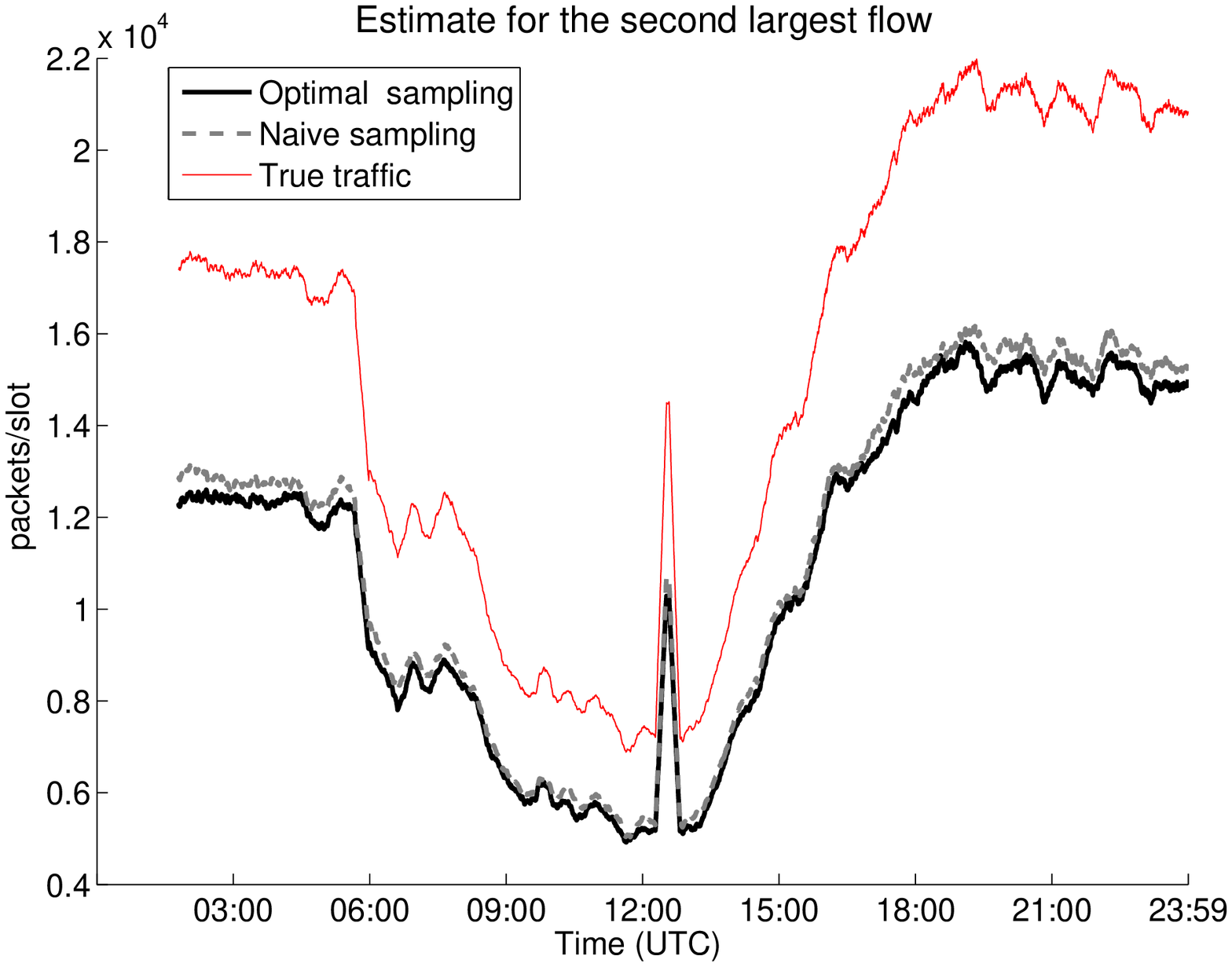}} 
    \hspace{-20pt}
    \subfigure[Third largest flow.]{\label{fig:opt_vs_naive_flow30}\includegraphics[scale=0.30]{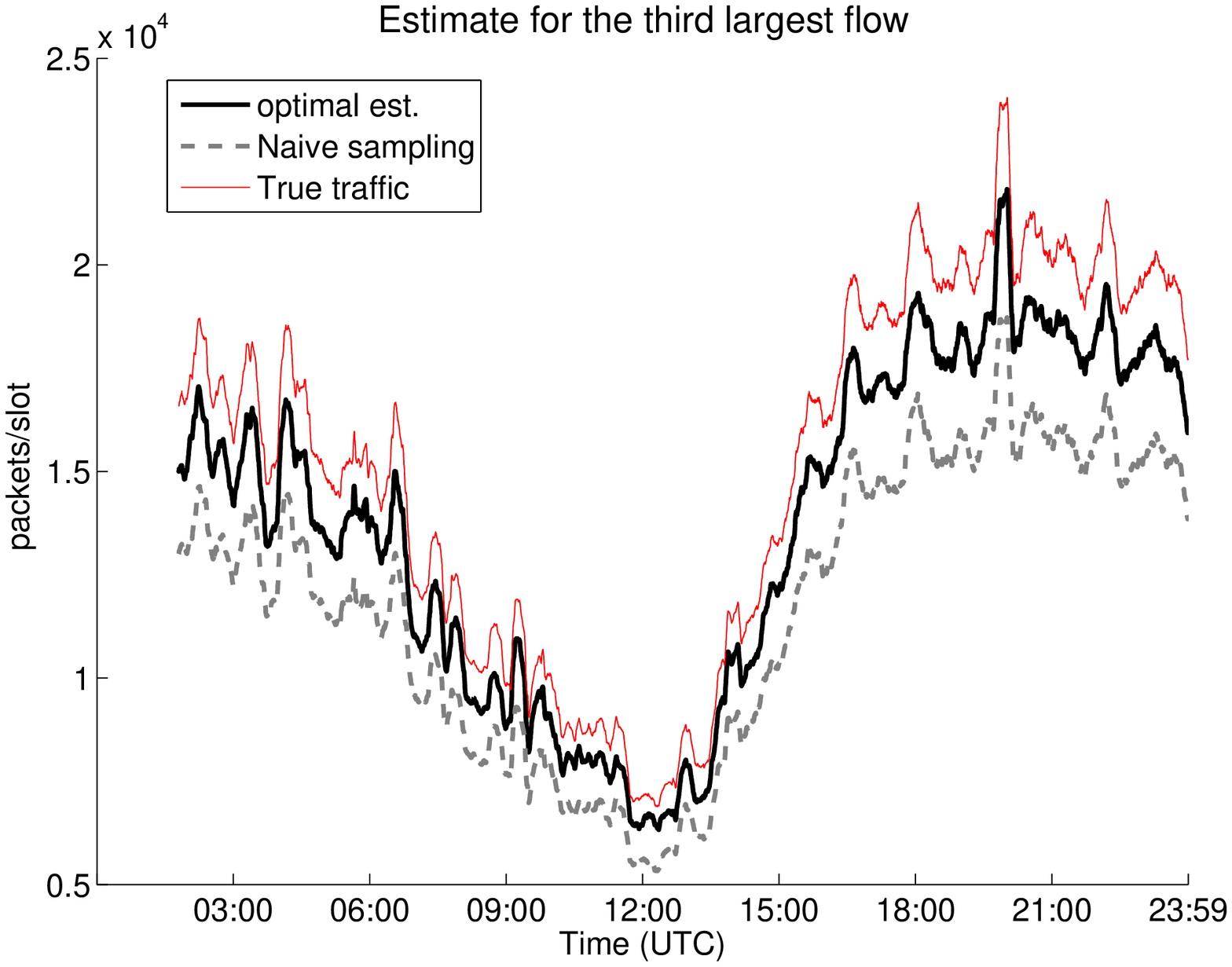}} 
  \end{center}
  \vspace{-10pt}
  \caption{Traffic estimation for different flows.}
  \label{fig:flow_estimation}
\end{figure*}

We use a real-world network, namely Internet2, to evaluate our algorithm.
We juxtapose our method against a na\"ive sampling scheme (i.e., sampling rates not chosen optimally; Kalman filtering is still used though).
Internet2   involves $L=26$ links, $N=9$ nodes and $J=72$ routes 
(see~\cite{stilianjoeltraffic,I2}). 
In particular, we employ a dataset for traffic captured on March 17, 2009.
The dataset includes the traffic volume of the $72$ flows, and the routing matrix $R$ (see Eq.~\eqref{e:routing}) which gives the path that each flow 
traverses in the network\footnote{All datasets used can be provided by the authors upon request.}.
In all examples that follow, a training window of $500$ time slots was applied to calibrate our model (see Step 1 of Algorithm~\ref{alg:kalman}).

In the na\"ive sampling scheme we evenly split the available sampling capacity among the competing flows of a link.
We assume that the sampling capacity for each of the $26$ network links is $0.20$. Figure~\ref{fig:opt_vs_naive1} shows the 
empirical root mean squared error (RMSE) for the whole network on the day of interest. RMSE is  defined as,
$
RMSE(t)=\sqrt{\sum_{j\in\mathcal{J}} ( \hat{x}_j(t|t) - x_j(t))^2/|J|}.
$
The RMSE time average for the optimal sampling scheme is $1548$ packets per time slot, and $1773$ packets per time slot for the na\"ive one.  This corresponds to a $13\%$ error reduction.

\begin{figure}[tp]
    \centering
    \includegraphics[scale=0.31]{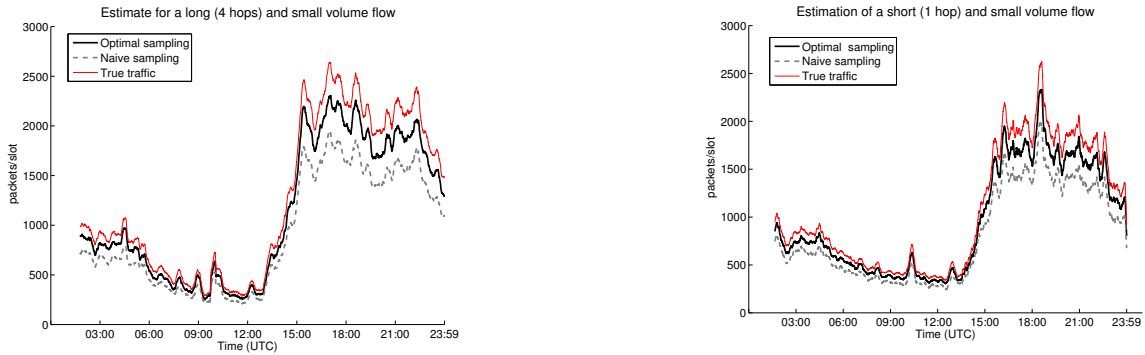}
    \vspace{-10pt}
   \caption{Long flow (4 hops), low  volume.}
   \label{fig:opt_vs_naive_flow14}
   \vspace{-10pt}
\end{figure}

We also examine traffic estimation for individual flows. Figures~\ref{fig:opt_vs_naive_flow19} -- \ref{fig:opt_vs_naive_flow30} present the cases
for the three largest flows in terms of average traffic volume size, namely flows $19$, $66$ and $30$.  
Moreover, Figure~\ref{fig:opt_vs_naive_flow14} illustrates the estimation outcome for flow $14$, which is a ``long" flow traversing $4$ links, but with relatively
low traffic volume. Similarly,  Figure~\ref{fig:opt_vs_naive_flow68}, depicts the results for flow $68$, a ``short" flow with low traffic volume.
Clearly, the proposed approached is advantageous over the simplistic sampling scheme.

The results indicate the performance gains of our sampling scheme, being a result of
considering  both temporal and spatial correlation between flows. A necessary requirement, though, is
the \emph{stationarity} of traffic volumes.  This does not always hold for Internet traffic. Ongoing work
includes investigation of ``richer" stochastic models, something that would allow sampling designs
with even stricter sampling constraints (e.g., $1:100$ or even $1:1000$). Furthermore,
one can additionally re-calibrate the model and ``learn" its new parameters by increasing the frequency of the training periods (step 1 of Algorithm~\ref{alg:kalman}).

\begin{figure}[tp]
    \centering
    \includegraphics[scale=0.31]{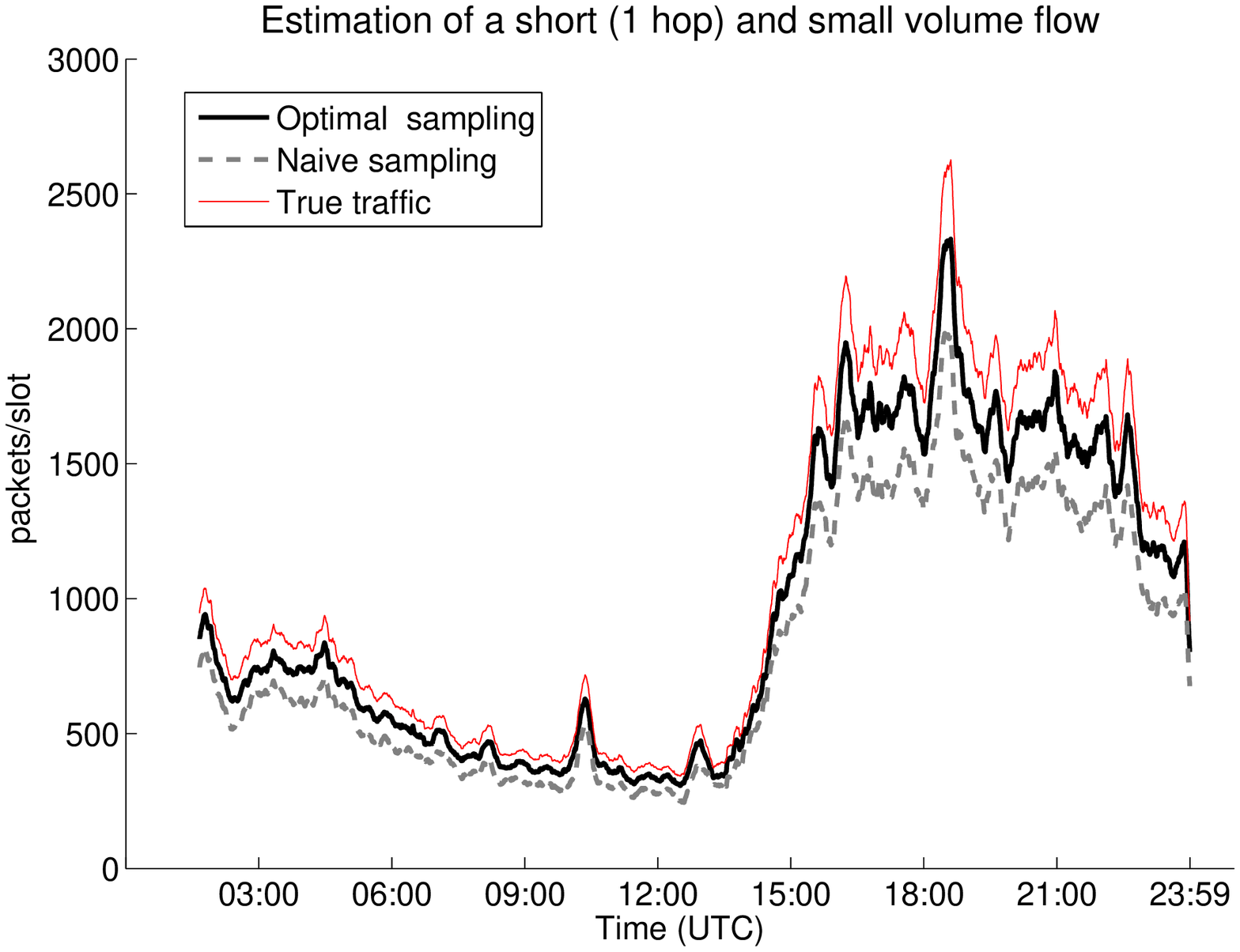}
        \vspace{-10pt}
    \caption{Short flow (1 hop), low volume.}
    \label{fig:opt_vs_naive_flow68}
        \vspace{-10pt}
\end{figure}

%
\bibliographystyle{abbrv}
\bibliography{sampling}  
%
%

\balancecolumns
\end{document}